\documentclass[showpacs,twocolumn,nofootinbib]{revtex4}
\usepackage{amsfonts}
\usepackage{graphicx}
\usepackage{amsmath}
\usepackage{amssymb}
\usepackage{amsthm}

\newtheorem{defn}{Definition}
\newtheorem{prop}{Proposition}
\newtheorem{cora}{Corollary}
\input{tcilatex}

\begin{document}

\title{The paradox of soft singularity crossing and its resolution \\
by distributional cosmological quantitities}
\author{Zolt\'{a}n Keresztes$^{1,2}$}
\email{zkeresztes@titan.physx.u-szeged.hu }
\author{L\'{a}szl\'{o} \'{A}. Gergely$^{1,2}$}
\email{gergely@physx.u-szeged.hu}
\author{Alexander Yu. Kamenshchik$^{3,4}$}
\email{Alexander.Kamenshchik@bo.infn.it}
\affiliation{$^{1}$Department of Theoretical Physics, University of Szeged, Tisza Lajos
krt 84-86, Szeged 6720, Hungary\\
$^{2}$ Department of Experimental Physics, University of Szeged, D\'{o}m T%
\'{e}r 9, Szeged 6720, Hungary\\
$^{3}$ Dipartimento di Fisica and INFN, via Irnerio 46, 40126 Bologna, Italy%
\\
$^{4}$ L. D. Landau Institute for Theoretical Physics, Russian Academy of
Sciences, Kosygin street 2, 119334 Moscow, Russia}

\begin{abstract}
A cosmological model of a flat Friedmann universe filled with a mixture of
anti-Chaplygin gas and dust-like matter exhibits a future soft singularity,
where the pressure of the anti-Chaplygin gas diverges (while its energy
density is finite). Despite infinite tidal forces the geodesics pass through
the singularity. Due to the dust component, the Hubble parameter has a
non-zero value at the encounter with the singularity, therefore the dust
implies further expansion. With continued expansion however, the energy
density and the pressure of the anti-Chaplygin gas would become ill-defined,
hence from the point of view of the anti-Chaplygin gas only a contraction is
allowed. Paradoxically, the universe in this cosmological model would have
to expand and contract simultaneously. This obviosly could not happen. We
solve the paradox by redefining the anti-Chaplygin gas in a distributional
sense. Then a contraction could follow the expansion phase at the
singularity at the price of a jump in the Hubble parameter. Although such an
abrupt change is not common in any cosmological evolution, we explicitly
show that the set of Friedmann, Raychaudhuri and continuity equations are
all obeyed both at the singularity and in its vicinity. We also prove that
the Israel junction conditions are obeyed through the singular spatial
hypersurface. In particular we enounce and prove a more general form of the
Lanczos equation.
\end{abstract}

\maketitle

\section{Introduction}

The problem of cosmological singularities has been attracting the attention
of theoreticians working in gravity and cosmology since the early fifties 
\cite{Land,Misn-Torn,Hawk-Ell}. In the sixties general theorems about the
conditions for the appearance of singularities were proven \cite{Hawk,Pen}
and the oscillatory regime of approaching the singularity \cite{BKL}, the
Mixmaster universe \cite{Misner} was discovered. Basically, until the end of
nineties almost all discussions about singularities were devoted to the Big
Bang and Big Crunch singularities, which are characterized by a vanishing
cosmological radius.

However, kinematical investigations of Friedmann cosmologies have raised the
possibility of sudden future singularity occurrence \cite{sudden1}-\cite%
{sudden9}, characterized by a diverging $\ddot{a}$ whereas both the scale
factor $a$ and $\dot{a}$ are finite. Then, the Hubble parameter $H=\dot{a}/a~
$and the energy density $\rho $ are also finite, while the first derivative
of the Hubble parameter and the pressure $p$ diverge. Until recent years,
however, sudden future singularities attracted only a limited interest among
researchers. The interest grew due to two reasons. The recent discovery of
the cosmic acceleration \cite{cosmic} has stimulated the elaboration of dark
energy models, responsible for such a phenomenon (see e.g. for review \cite%
{dark}). Remarkably in some of these models the sudden singularities arise
quite naturally. Another source of the interest to sudden singularities is
the development of brane models \cite{Shtanov,Shtanov1,sudden9}, where
singularities of this kind could arise naturally (sometimes these
singularities, arising in brane-world models, are called \textquotedblleft
quiescent\textquotedblright\ \cite{Shtanov}).

In the investigations devoted to sudden singularities one can distinguish
three main topics. The first of them deals with the question of the
compatibility of the models possessing soft singularities with observational
data \cite{sudden6,Mar,tach1}. The second direction is connected with the
study of quantum effects \cite{Shtanov1,sudden8,Haro,quantum,quantum1}. Here
one can see two subdirections: the study of quantum corrections to the
effective Friedmann equation, which can eliminate classical singularities or
at least, change their form \cite{Shtanov,sudden8,Haro}; and the study of
solutions of the Wheeler-DeWitt equation for the quantum state of the
universe in the presence of sudden singularities \cite{quantum,quantum1}.
The third direction is connected with the possibility of the sudden
singularity crossing in classical cosmology \cite%
{Lazkoz,Lazkoz1,BalDab,tach2,quantum1}. The present paper is devoted exactly
to this topic.

A particular feature of the sudden future singularities is their softness 
\cite{Lazkoz}. As the Christoffel symbols depend only on the first
derivative of the scale factor, they are regular at these singularities.
Hence, the geodesics are well behaved and they can cross the singularity 
\cite{Lazkoz}. One can argue that the particles crossing the singularity
will generate the geometry of the spacetime, providing in such a way a
\textquotedblleft soft rebirth\textquotedblright\ of the universe after the
singularity crossing \cite{tach2}. Note that the possibility of crossing of
some kind of cosmological singularities was noticed already in the early
paper by Tipler \cite{Tipler}. A close idea of integrable singularities in
black holes, which can give origin to a cosmogenesis was recently put
forward in \cite{Lukash}.

As a starting point we consider an interesting example of a sudden future
singularity - the Big Brake which was discovered in Ref. \cite{tach0} while
studying a particular tachyon cosmological model. The particularity of the
Big Brake singularity consists in the fact that the time derivative of the
scale factor is not only finite, but exactly equal to zero. That makes the
analysis of the behavior in the vicinity of singularity especially
convenient. In particular, in Ref. \cite{tach1} it was shown that the
predictions of the future of the universe in this model \cite{tach0} are
compatible with the supernovae type Ia data, while in Refs. \cite%
{quantum,quantum1} some quantum cosmological questions were studied in the
presence of the Big Brake singularity.

The simplest cosmological model allowing a Big Brake singularity was also
introduced in Ref. \cite{tach0}. This model is based on the perfect fluid,
dubbed \textquotedblleft anti-Chaplygin gas\textquotedblright . This fluid
is characterized by the equation of state 
\begin{equation}
p=\frac{A}{\rho },  \label{anti-Chap}
\end{equation}%
where $A$ is a positive constant. Such and equation of state arises, for
example, in the theory of wiggly strings \cite{wiggy}. In paper \cite{tach0}
a fluid obeying the equation of state (\ref{anti-Chap}) was called
"anti-Chaplygin gas" in analogy with the Chaplygin gas \cite{Chap} which has
the equation of state $p=-A/\rho $ and has acquired some popularity as a
candidate for a unified theory of dark energy and dark matter \cite{Chap1}.

An explicit example of the crossing of the Big Brake singularity was
described in detail in paper \cite{tach2}, were the tachyon model \cite%
{tach0,tach1} was investigated. In this model the tachyon field passes
through the singularity, continuing its evolution with a recollapse towards
a Big Crunch. In a simpler model, based on the anti-Chaplygin gas, such a
crossing is even easier to understand.

The next natural step in the analysis of the soft singularities seems to be
obvious. One can consider a soft singularity of more general type than the
big brake by adding to the tachyon matter or to the anti-Chaplygin gas some
dustlike matter. However, in this case the traversability of the singularity
seems to be obstructed. The main reason for this is that while the energy
density of the tachyonic field (or of the anti-Chaplygin gas) vanishes at
the singularity, the energy density of the matter component does not,
leaving the Hubble parameter at the singularity with a finite value. Then
some kind of the paradox arises: if the universe continues its expansion,
and if the equation of state of the component of matter, responsible for the
appearance of the soft singularity (in the simplest case, the anti-Chaplygin
gas) is unchanged, then the expression for the energy density of this
component becomes imaginary, which is unacceptable. The situation looks
rather strange: indeed, the model, including dust should be in some sense
more regular, than a single exotic fluid, the anti-Chaplygin gas. Thus, if
the model based on the pure anti-Chaplygin gas has a traversable Big Brake
singularity, than the more general singularity arising in the model based on
the mixture of the anti-Chaplygin gas and dust should also be traversable.

Related to that, it was recently shown that general soft singularities
arising in the Friedmann model, filled with the scalar field with a negative
potential, inversely proportional to this field are traversable. So, what
could be wrong with the simple two-fluid model? One can see that what we
face is some sort of a clash between the equation of state of one of these
fluids and the dynamics (the Friedmann and Raychaudhuri equations) and
energy conservation equations. In this paper we shall try to resolve this
paradox, insisting on the preservation of the equation of state of the
anti-Chaplygin gas. The price which one has to pay for it is the obligatory
use of the generalized functions for some cosmological quantities. Namely,
the anti-Chaplygin gas remains physical if rather a recollapse follows, but
then the Hubble parameter would have a sharp jump, obstructing the validity
of the Raychaudhuri equation (the second Friedmann equation) in the usual
sense of functions. Thus, apparently, the evolution cannot be continued
through the soft singularity, unless treating the cosmological quantities as
distributions. We claim that such a generalization is mathematically
rigorous, moreover, the introduction of distributions is not so drastic, as
it looks at the first sight, as the pressure of the anti-Chaplygin gas
diverges anyhow at the soft singularity (as it so does for the tachyon
field). Then in the Conclusion we shall dwell on the possible physical sense
of the proposed constructions and its possible alternatives.

The plan of the paper is as follows. In Section \ref{FriedmannFlat} we
discuss generic Friedmann space-times, which admit $\dot{H}=-\infty $ type
singularities, while the Hubble parameter $H$ remains finite. Such
singularities are related to corresponding divergencies in the pressure of
the perfect fluid filling the Friedmann universe (while its energy density
stays finite). We investigate the kinematics, the geodesic equations, the
geodesic deviation equations in the vicinity of these singularities and also
prove that these singularities are weak.

In Section \ref{antiChap} we discuss a mixture of the anti-Chaplygin gas and
dust in a flat Friedmann universe and explain the essence of the paradox. We
explicitly derive the behavior of the energy density and pressure in the
vicinity of the soft singularity and we solve the geodesic equations in this
region. While the singularity turns to be traversable by the geodesics, the
explicit solution also shows that the Raychaudhuri equation is violated at
the singularity.

In Section \ref{DistChap} we add generalized distributional contributions to
both the pressure and energy density, such that (a) the equation of state of
the anti-Chaplygin gas still holds and (b) the singularity becomes
traversable. We also perform checks of the Friedmann, Raychaudhuri and
continuity equations, which all hold valid across the singularity in a
distributional sense. In the process we employ a number of Propositions on
distributions presented and proved in Appendix A. For the convenience of the
reader we present a related semi-heuristic discussion of two known
distributional identities in Appendix B. We stress that the distributional
modifications of the energy density and pressure do not modify the
cosmological evolution, but they make possible the soft singularity crossing.

In Section \ref{Junction} we revisit the junction conditions along a
spacelike hypersurface in a flat Friedmann universe. The future soft
singularity represents such a spatial hypersurface, along which the
energy-momentum tensor diverges. Extending the space-time through this
hypersurface is possible by obeying both Israel junction conditions. While
the first condition, requiring the continuity of the induced metric is easy
to satisfy (the metric stays regular at the soft singularity), the second
condition relates the jump in the extrinsic curvature to the distributional
part of the energy-momentum tensor through the Lanczos equation. We will
show that in flat Friedmann space-times the Lanczos equation holds for a
more general class of \textit{distributional} energy-momentum tensors. With
this we give a second proof that the generalized distributional
energy-momentum tensor assures the traversability of the soft singularity.
In the process we employ a simple form of the Lanczos equation valid in flat
Friedmann universes, derived in Appendix \ref{Lanczos}.

We summarize our results and give some further outlook in the Concluding
Remarks.

We chose $c=1$ and $8\pi G/3=1$. A subscript $S$ denotes the value of the
respective quantity at the soft singularity.

\section{Pressure singularities in flat Friedmann universes\label%
{FriedmannFlat}}

The line element squared of a flat Friedmann universe can be written as%
\begin{equation}
ds^{2}=-dt^{2}+a^{2}\left( t\right) \sum_{\alpha }\left( dx^{\alpha }\right)
^{2}\ ,
\end{equation}%
where $x^{\alpha }$ ($\alpha =1,2,3$) are Cartesian coordinates. The
evolution of the Friedmann universe is governed by the Raychaudhuri (second
Friedmann) equation%
\begin{equation}
\dot{H}=-\frac{3}{2}\left( \rho +p\right) \ ,  \label{Raych}
\end{equation}%
and by the continuity equation for the fluid 
\begin{equation}
\dot{\rho}+3H\left( \rho +p\right) =0\ .  \label{cont}
\end{equation}%
Here the dot denotes the derivative with respect to cosmological time $t$. A
first integral of this system is given by the first Friedmann equation%
\begin{equation}
H^{2}=\rho \ .  \label{Friedmann}
\end{equation}%
It is easy to see that the Raychaudhuri equation can be obtained from the
first Friedmann and the continuity equations.

\subsection{Kinematics in the vicinity of sudden singularities\label%
{Kinematics}}

Sudden singularities are characterized by finite $H$ and $\dot{H}\rightarrow
-\infty $ (finite $\dot{a}$ and $\ddot{a}\rightarrow -\infty $) at some
finite scale factor $a$. The energy density of the fluid is finite but its
pressure diverges at this type of singularity, therefore the term
\textquotedblleft pressure singularity\textquotedblright\ is also in use.
Then, we would like to emphasize the fact that there is an essential
difference between the sudden singularities with $H=0$ and with $H>0$. As
has been already mentioned in the Introduction, in the first case, which is
called Big Brake, the universe begin contracting and running towards the Big
Crunch singularity. Exactly this occurs in models based on tachyon field
with a particular potential \cite{tach0,tach1,tach2} or in the
anti-Chaplygin gas models. In the case of the model based on the mixture of
one of this fluids and dust, we encounter the second situation when the
value of the Hubble constant is positive at the moment of encounter with the
sudden singularity. That means that after crossing the singularity the
universe should continue its expansion, but the anti-Chaplygin gas becomes
ill-defined, as it will be shown in detail in Section \ref{antiChap},
devoted to the model based on mixture of the anti-Chaplygin gas and dust.

One possible way of overcoming this obstacle is to allow the jump in the
sign of the Hubble parameter, which as was mentioned in the Introduction
leaves valid the first Friedmann equation, the continuity equation and the
equation of state, while making invalid the Raychaudhuri equation. This last
obstacle can be cured by the acceptance the distributional Dirac $\delta $%
-function type contributions into the pressure and the energy density, which
is described in detail in the section IV.

\subsection{Geodesics in the vicinity of sudden singularities\label%
{geodesics}}

The geodesic equations in flat Friedmann space-time are%
\begin{equation}
\frac{d^{2}x^{\alpha }}{d\lambda ^{2}}+2\frac{\dot{a}}{a}\frac{dt}{d\lambda }%
\frac{dx^{\alpha }}{d\lambda }=0\ ,
\end{equation}%
\begin{equation}
\frac{d^{2}t}{d\lambda ^{2}}+a\dot{a}\sum_{\alpha }\left( \frac{dx^{\alpha }%
}{d\lambda }\right) ^{2}=0\ ,
\end{equation}%
where $\lambda $ is an affine parameter. Integrating these equations yields%
\begin{equation}
\frac{dx^{\alpha }}{d\lambda }=\frac{P^{\alpha }}{a^{2}}\ ,
\label{xadifflambda}
\end{equation}%
\begin{equation}
\left( \frac{dt}{d\lambda }\right) ^{2}=\epsilon +\frac{P^{2}}{a^{2}}\ ,
\label{tdifflambda}
\end{equation}%
with $P^{\alpha }$, $\epsilon $ integration constants and $%
P^{2}=\sum_{\alpha }\left( P^{\alpha }\right) ^{2}$. The quantity $\epsilon $
is fixed by the length of the tangent vector $u^{a}$ of the geodesic as $%
\epsilon =-u_{a}u^{a}$; i.e. one for timelike and zero for lightlike orbits.
In a comoving system $P^{\alpha }=0$ and $t=\lambda $ is affine parameter.

Eqs. (\ref{xadifflambda}) and (\ref{tdifflambda}) are singular only for
vanishing scale factor (see also Ref. \cite{Lazkoz}). Therefore, the
existence of a solution $t\left( \lambda \right) $, $x^{\alpha }\left(
\lambda \right) $ of Eqs. (\ref{xadifflambda}) and (\ref{tdifflambda}) is
assured by the Cauchy-Peano theorem for any nonzero $a$ (including the soft
singularity). Thus the functions $t\left( \lambda \right) $ and $x^{\alpha
}\left( \lambda \right) $, i.e. the geodesics can be continued through the
singularity occurring at finite scale factor. Only derivatives of higher
order than two of $t\left( \lambda \right) $ and $x^{\alpha }\left( \lambda
\right) $ are singular (as they contain $\ddot{a}$), however these do not
appear in the geodesic equations. Pointlike particles moving on geodesics do
not experience any singularity. Thus, as we argued in the preceding paper 
\cite{tach2} one is not obliged to consider such a singularity as a final
state of the universe. Indeed, passing through this singularity the matter
recreates also the spacetime in a unique way, at least for such simple
models, as those based on Friedmann metrics.

\subsection{Deviation equation in the vicinity of sudden singularities}

The 3-spaces with $t=$const have vanishing Riemann curvature. However, the
4-dimensional Riemann curvature tensor has the nonvanishing components:%
\begin{eqnarray}
R_{\,\,\,\,\,t\beta t}^{\alpha } &=&-\frac{\ddot{a}}{a}\delta _{\beta
}^{\alpha }=\left( -\dot{H}+H^{2}\right) \delta _{\beta }^{\alpha }\ , 
\notag \\
R_{\,\,\,\,\,212}^{1} &=&R_{\,\,\,\,\,313}^{1}=R_{\,\,\,\,\,323}^{2}=\dot{a}%
^{2}  \label{Riemann}
\end{eqnarray}%
and the corresponding components arising from symmetry. Here $\alpha ,\beta
=1,2,3$. Remarkably, all components which diverge at the singularity are of
the type $R_{tata}$ \cite{tach2}. Therefore, the singularity arises in the
mixed spatio-temporal components.

The geodesic deviation equation along the integral curves of $u=\partial
/\partial t$ (which are geodesics with affine parameter $t$) is%
\begin{equation}
\dot{u}^{a}=-R_{\ \ cbd}^{a}\eta ^{b}u^{c}u^{d}\ ,
\end{equation}%
where $\eta ^{b}$ is the deviation vector separating neighboring geodesics,
chosen to satisfy $\eta ^{b}u_{b}=0$. For a Friedmann universe it becomes 
\begin{equation}
\dot{u}^{a}=-R_{\ \ tbt}^{a}\eta ^{b}\propto \ddot{a}\ ,
\end{equation}%
which at the singularity diverges as $-\infty $. Therefore, when approaching
the singularity, the tidal forces manifest themselves as an infinite braking
force stopping the further increase of the separation of geodesics, but not
the evolution along the geodesics. With $\ddot{a}<0$ in the vicinity of the
singularity, once the geodesics have passed through, they will approach each
other. Therefore a contraction phase will follow: everything that has
reached the singularity will bounce back.

\subsection{The type of the singularity\label{singNature}}

In this subsection we shall present the classification of singularities,
based on the point of view of finite size objects, which approach these
singularities. In principle, finite size objects could be destroyed while
passing through the singularity due to the occurring infinite tidal forces.
A strong curvature singularity is defined by the requirement that an
extended finite object is crushed to zero volume by tidal forces. We give
below Tipler's \cite{Tipler} and Kr\'{o}lak's \cite{Krolak} definitions of
strong curvature singularities together with the relative necessary and
sufficient conditions. An alternative definition of the softness of a
singularity, based on a Raychaudhuri averaging, was developed by Dabrowski 
\cite{Dab}.

According to Tipler's definition if every volume element defined by three
linearly independent, vorticity-free, geodesic deviation vectors along every
causal geodesic through a point $p$ vanishes, a strong curvature singularity
is encountered at the respective point $p$ \cite{Tipler}, \cite{Lazkoz}. The
necessary and sufficient condition for a causal geodesic to run into a
strong singularity at $\lambda _{s}$ ($\lambda $ is affine parameter of the
curve) \cite{ClarkeKrolak} is that the double integral%
\begin{equation}
\int_{0}^{\lambda }d\lambda ^{\prime }\int_{0}^{\lambda ^{\prime }}d\lambda
^{\prime \prime }\left\vert R_{\,\,ajb}^{i}u^{a}u^{b}\right\vert
\label{cond1}
\end{equation}%
diverges as $\lambda \rightarrow \lambda _{s}$. A similar condition is valid
for lightlike geodesics, with $R_{\,\,ajb}^{i}u^{a}u^{b}$ replacing $%
R_{\,ab}u^{a}u^{b}$ in the double integral.

Kr\'{o}lak's definition is less restrictive. A future-endless,
future-incomplete null (timelike) geodesic $\gamma $ is said to terminate in
the future at a strong curvature singularity if, for each point $p\in \gamma 
$, the expansion of every future-directed congruence of null (timelike)
geodesics emanating from $p$ and containing $\gamma $ becomes negative
somewhere on $\gamma $ \cite{Krolak}, \cite{GenStrong}. The necessary and
sufficient condition for a causal geodesic to run into a strong singularity
at $\lambda _{s}$ \cite{ClarkeKrolak} is that the integral%
\begin{equation}
\int_{0}^{\lambda }d\lambda ^{\prime }\left\vert
R_{\,\,ajb}^{i}u^{a}u^{b}\right\vert  \label{cond2}
\end{equation}%
diverges as $\lambda \rightarrow \lambda _{s}$. Again, a similar condition
is valid for lightlike geodesics, with $R_{\,\,ajb}^{i}u^{a}u^{b}$ replacing 
$R_{\,ab}u^{a}u^{b}$ in the integral.

In flat Friedmann space-time the comoving observers move on geodesics having
four velocity $u=\partial /\partial t$, where $t$ is affine parameter. The
nonvanishing components of Riemann tensor are given by Eq. (\ref{Riemann}).
Since $H$ is finite along the geodesics, neither of the integrals (\ref%
{cond1}) and (\ref{cond2}) diverge at the singularity. The singularity is
weak (soft) according to both Tipler's and Kr\'{o}lak's definitions. That
means although the tidal forces become infinite, the finite objects are not
necessarily crushed when reaching the singularity (see also \cite{Lazkoz}).

\section{The paradox of the soft singularity crossing in the cosmological
model based on the anti-Chaplygin gas and dust universe\label{antiChap}}

We discuss an universe filled with two components. One is the anti-Chaplygin
gas with the equation of state (\ref{anti-Chap}) and other is the
pressureless dust.

The solution of the continuity equation for the anti-Chaplygin gas gives the
following dependence of its energy density on the scale factor: 
\begin{equation}
\rho _{ACh}=\sqrt{\frac{B}{a^{6}}-A}\ ,  \label{anti-Chap1}
\end{equation}%
where $B$ is a positive constant, determining the initial condition. The
energy density of the dust-like matter is as usual 
\begin{equation}
\rho _{m}=\frac{\rho _{m,0}}{a^{3}}\ ,~  \label{dust}
\end{equation}%
where $\rho _{m,0}$ is a constant.

It is clear that when during the expansion of the universe, its scale factor
approaches the value 
\begin{equation}
a_{S}=\left( \frac{B}{A}\right) ^{\frac{1}{6}}\ ,  \label{BB}
\end{equation}%
the energy density of the anti-Chaplygin gas vanishes, and its pressure
grows to infinity. That means that the deceleration also becomes infinite.
However, the energy density of dust remains finite, hence the same is true
also for the Hubble parameter. It is here that the paradox arises: if the
universe continues to expand, the expression under the sign of the square
root in Eq. (\ref{anti-Chap1}) becomes negative and the energy density of
the anti-Chaplygin gas becomes ill-defined. A way out of this situation is
only by assuming that at this moment the Hubble parameter changes its sign,
while keeping its absolute value (such that the energy density will not have
a jump, as implied by the Friedmann equation). This possibility will be
studied in detail in the following subsections.

\subsection{Evolutions in the vicinity of the singularity\label{closeFSS}}

Let us substitute the expressions (\ref{anti-Chap1}) and (\ref{dust}) into
the first Friedmann equation. We shall find its solution for the universe
approaching to the soft singularity point at the moment $t_{S}$ (the latter
cannot be found analytically, but its value is not important for our
analysis): 
\begin{equation}
a(t)=a_{S}-\sqrt{\frac{\rho _{m,0}}{a_{S}}}(t_{S}-t)-\sqrt{\frac{2Aa_{S}^{2}%
}{3H_{S}}}(t_{S}-t)^{3/2}\ ,  \label{Fried-sing}
\end{equation}%
where%
\begin{equation}
H_{S}=\sqrt{\frac{\rho _{m,0}}{a_{S}^{3}}}  \label{Hubble-sing}
\end{equation}%
is the value of the Hubble parameter at $t=t_{S}$. Correspondingly the
leading terms of the energy densities of the anti-Chaplygin gas and dust,
also of the pressure of the anti-Chaplygin gas are 
\begin{equation}
\rho _{m}=H_{S}^{2}+3H_{S}^{3}(t_{S}-t)\ ,  \label{dust1}
\end{equation}%
\begin{equation}
\rho _{ACh}=\sqrt{6AH_{S}(t_{S}-t)}\ ,  \label{anti-Chap2}
\end{equation}%
\begin{equation}
p_{ACh}=\sqrt{\frac{A}{6H_{S}(t_{S}-t)}}\ .  \label{anti-Chap3}
\end{equation}

One can see that the expressions (\ref{Fried-sing}), (\ref{anti-Chap2}) and (%
\ref{anti-Chap3}) cannot be continued for $t>t_{S}$ due to the emerging
negative quantities under the square roots. The assumption of a sharp
transition from expansion to contraction implies the following changes in
Eqs. (\ref{Fried-sing}) -- (\ref{anti-Chap3}): 
\begin{equation}
a(t)=a_{S}-\sqrt{\frac{\rho _{m,0}}{a_{S}}}|t_{S}-t|-\sqrt{\frac{2Aa_{S}^{2}%
}{3H_{S}}}|t_{S}-t|^{3/2}\ ,  \label{Fried-sing1}
\end{equation}%
\begin{equation}
\rho _{m}=H_{S}^{2}+3H_{S}^{3}|t_{S}-t|\ ,  \label{dust2}
\end{equation}%
\begin{equation}
\rho _{ACh}=\sqrt{6AH_{S}|t_{S}-t|}\ ,  \label{anti-Chap20}
\end{equation}%
\begin{equation}
p_{ACh}=\sqrt{\frac{A}{6H_{S}|t_{S}-t|}}\ .  \label{anti-Chap30}
\end{equation}%
The quantities (\ref{Fried-sing1})--(\ref{anti-Chap20}) are well-defined and
continuous at the moment of the singularity crossing. The expression for the
pressure (\ref{anti-Chap30}) is divergent, but this divergence is integrable
and this is sufficient for our purposes. These new expressions satisfy the
Friedmann equation, the continuity equations and the equation of state for
the anti-Chaplygin gas. However, the time derivatives of these quantities
are not continuous and it is the reason of the failure of the Raychaudhuri
equation. We shall analyze this problem in the following section, but before
we discuss the geodesics in the vicinity of the singularity.

\subsection{Singularity crossing geodesics}

We can integrate explicitly the geodesics equations (\ref{xadifflambda}) and
(\ref{tdifflambda}) in the vicinity of singularity, using the expression (%
\ref{Fried-sing1}) for the cosmological factor, also taken in the vicinity
of singularity. Choosing the affine parameter in such a way that the point $%
\lambda =0$ corresponds to the singularity crossing we obtain up to the
second order in $\lambda $ terms 
\begin{equation}
t=t_{S}+\sqrt{\epsilon +\frac{P^{2}}{a_{S}^{2}}}\lambda +\frac{P^{2}H_{S}}{%
2a_{S}^{2}}sgn(\lambda )\lambda ^{2}\ ,  \label{tgeod}
\end{equation}%
\begin{equation}
x^{\alpha }=x_{S}^{\alpha }+\frac{P^{\alpha }}{a_{S}^{2}}\lambda +\sqrt{%
\epsilon +\frac{P^{2}}{a_{S}^{2}}}\frac{P^{\alpha }H_{S}}{a_{S}^{2}}%
sgn(\lambda )\lambda ^{2}\ .  \label{xgeod}
\end{equation}%
One can see from Eqs. (\ref{tgeod}) and (\ref{xgeod}) that not only the time
and spatial coordinates of the geodesics are continuous at the soft
singularity crossing, but also their first derivatives with respect to the
affine parameter $\lambda $.

\section{Singularity crossing, the Raychaudhuri equation and distributions 
\label{DistChap}}

Let us discuss the expressions for the Hubble parameter and its time
derivative in the vicinity of the singularity. Starting from the expression (%
\ref{Fried-sing1}) we obtain 
\begin{eqnarray}
H(t) &=&H_{S}sgn(t_{S}-t)  \notag \\
&&+\sqrt{\frac{3A}{2H_{S}a_{S}^{4}}}sgn(t_{S}-t)\sqrt{|t_{S}-t|}~,
\label{Hubble-sing2}
\end{eqnarray}%
\begin{equation}
\dot{H}=-2H_{S}\delta (t_{S}-t)-\sqrt{\frac{3A}{8H_{S}a_{S}^{4}}}\frac{%
sgn(t_{S}-t)}{\sqrt{|t_{S}-t|}}~.  \label{Hubble-der-sing}
\end{equation}%
Naturally, the $\delta $-term in $\dot{H}$ arises because of the jump in $H$%
, as the expansion of the universe is followed by a contraction. To restore
the validity of the Raychaudhuri equation we shall add a singular $\delta $%
-term to the pressure of the anti-Chaplygin gas, which will acquire the form 
\begin{equation}
p_{ACh}=\sqrt{\frac{A}{6H_{S}|t_{S}-t|}}+\frac{4}{3}H_{S}\delta (t_{S}-t)~.
\label{pressure-new}
\end{equation}%
The equation of state (\ref{anti-Chap}) of the anti-Chaplygin gas is
preserved, if we also modify the expression for its energy density: 
\begin{equation}
\rho _{ACh}=\frac{A}{\sqrt{\frac{A}{6H_{S}|t_{S}-t|}}+\frac{4}{3}H_{S}\delta
(t_{S}-t)}~.  \label{en-den-new}
\end{equation}%
The last expression should be understood in the sense of the composition of
distributions (see Appendix A and the references therein).

In order to prove that $p_{ACh}$ and $\rho _{ACh}$ represent a
self-consistent solution of the system of cosmological equations, we shall
use the following distributional identities: 
\begin{eqnarray}
\left[ sgn\left( \tau \right) g\left( \left\vert \tau \right\vert \right) %
\right] \delta \left( \tau \right) &=&0\ ,  \label{proposition1} \\
\left[ f\left( \tau \right) +C\delta \left( \tau \right) \right] ^{-1}
&=&f^{-1}\left( \tau \right) \ ,  \label{proposition2} \\
\frac{d}{d\tau }\left[ f\left( \tau \right) +C\delta \left( \tau \right) %
\right] ^{-1} &=&\frac{d}{d\tau }f^{-1}\left( \tau \right) \ .
\label{proposition3}
\end{eqnarray}%
Here $g\left( \left\vert \tau \right\vert \right) $ is bounded on every
finite interval, $f\left( \tau \right) >0$ and $C>0$ is a constant. These
identities follow from the Propositions 1, 2 and the Corollary enounced and
proved in the Appendix A. The parameter $\tau $ stays instead of the
difference $t_{S}-t$.

Because of Eq. (\ref{proposition2}), the energy density (\ref{en-den-new})
behaves as a continuous function which vanishes at the singularity. The
first term in the expression for the pressure (\ref{pressure-new}) diverges
at the singularity. Therefore the addition of a Dirac delta term, which is
not changing the value of $p_{ACh}$ at any $\tau \neq 0$ (i.e. $t\neq t_{S}$%
) does not look too drastic and might be considered as a some kind of
renormalization.

To prove that Friedmann, Raychaudhuri and continuity equations are satisfied
we must only investigate those terms, appearing in the field equations,
which contain Dirac $\delta$-functions, since without them, these equations
can be reduced to those we have found in the previous section. First, we
check the continuity equation for the anti-Chaplygin gas. Due to the
identities (\ref{proposition2})-(\ref{proposition3}), the $\delta \left(
\tau \right) $-terms occurring in $\rho _{ACh}$ and $\dot{\rho}_{ACh}$ could
be dropped. We keep them however in order to have the equation of state
explicitly satisfied. Then the $\delta \left( \tau \right) $-term appearing
in $3Hp_{ACh}$ vanishes, because the Hubble parameter changes sign at the
singularity [see Eq. (\ref{proposition1})].

The $\delta \left( \tau \right) $-term appearing in $\rho _{ACh}$ does not
affect the Friedmann equation due to the identity (\ref{proposition2}).
Finally, the $\delta $-term arising in the time derivative of the Hubble
parameter in the left-hand side of the Raychaudhuri equation is compensated
by the conveniently chosen $\delta $-term in the right-hand side of Eq. (\ref%
{pressure-new}).

\section{The junction conditions across the singularity\label{Junction}}

In this section we discuss the singularity crossing in a slightly different
way, by analyzing the junction conditions. We have to match two space-time
regions across the space-like hypersurface $\tau =0$. The junction of two
space-time regions has to obey the Israel matching conditions \cite{Israel},
namely, the induced metric should be continuous and the extrinsic curvature
of the junction hypersurface could possibly have a jump, which is related to
the distributional energy-momentum tensor on the hypersurface by the Lanczos
equation. The scale factor being continuous across the singularity the first
Israel condition is obeyed. We will next prove that the second Israel
junction condition (the Lanczos equation \cite{Lanczos}, \cite{Israel}) is
also satisfied.

For this we have to check whether Eq. 
\begin{equation}
\frac{\partial }{\partial t}\left( Ha^{2}\right) =\left\{ -H^{2}+\frac{3}{2}%
\left[ \widetilde{\rho }-\widetilde{p}-\overline{p}\delta \left( \tau
\right) \right] \right\} a^{2}\ ,  \label{Hderiv1}
\end{equation}%
(see Appendix \ref{Lanczos}), still implies the Lanczos equation 
\begin{equation}
\Delta H|_{t_{s}}=-\frac{3}{2}\overline{p}\ ,  \label{L1}
\end{equation}%
derived in Appendix \ref{Lanczos}, when $\widetilde{p}+\overline{p}\delta
\left( \tau \right) =p_{ACh}$, $\widetilde{\rho }=\rho _{m}+\rho _{ACh}$ and 
$\rho _{ACh}$ is generalized to a distribution 
\begin{equation}
\rho _{ACh}a^{2}=\frac{P\left( \tau \right) }{\left[ R\left( \tau \right)
+Q\left( \tau \right) \delta \left( \tau \right) \right] ^{\omega }}~.
\end{equation}
Here $\omega >0$, $R\left( \tau \right) >0$, $Q\left( \tau \right) >0$ and $%
P\left( \tau \right) $ is bounded.

When Eq. (\ref{Hderiv1}) is applied to a test function $\varphi \left( \tau
\right) $, the terms containing $H^{2}$ and $\rho _{m}$ give regular
contributions and the limits of the respective integrals vanish, similarly
as discussed in Appendix \ref{Lanczos}. Also, due to Proposition 2 given in
the Appendix A, the integral of the distributional term containing $\rho
_{ACh}$ becomes%
\begin{equation}
\int_{-\varepsilon }^{\varepsilon }\frac{P\left( \tau \right) \varphi \left(
\tau \right) }{R^{\omega }\left( \tau \right) }d\tau ~,
\end{equation}%
which also vanishes for $\varepsilon \rightarrow 0$. We still have to
consider the contributions 
\begin{equation}
\int_{-\varepsilon }^{\varepsilon }\left[ \widetilde{p}+\overline{p}\delta
\left( \tau \right) \right] \varphi \left( \tau \right) a^{2}d\tau \ .
\end{equation}%
Although the contribution $\widetilde{p}\varphi \left( \tau \right) a^{2}$
to the integrand is singular at $\tau =0$, its integral can be conveniently
evaluated by the Residue Theorem. For this we remark, that the integrand is
an analytically extendible function into the complex plane in the vicinity
of $\tau =0$ and its residue is zero, therefore the integral vanishes.
Finally, the contribution containing $\overline{p}$ leads by integration and
the limiting process to the right hand side of the Lanczos equation (\ref{L1}%
).

Therefore we have proven that the space-time regions separated by the
singular spatial hypersurface, representing the pressure singularity, can be
joined. In other words, the singularity becomes traversable.

\section{Concluding Remarks}

It is known that certain models of cosmological fluids in Friedmann
universes, like the anti-Chaplygin gas or the tachyon field with a special
potential \cite{tach0}, evolve into a sudden future singularity, which in
spite of a diverging pressure, is weak. It was argued that singularities of
this kind could be traversable despite infinite tidal forces emerging at the
singularity for an infinitesimally short time \cite{Lazkoz}. In Ref. \cite%
{tach2} the process of crossing of the Big Brake singularity was described
in some detail for the tachyon model \cite{tach0}. (The particularity of the
Big Brake singularity, consists in the fact that at the crossing of such a
singularity the Hubble variable is not only finite, but vanishes.) We also
note recent discussions \cite{nonsoft} on crossing the \textquotedblleft
traditional\textquotedblright\ Big Bang and Big Crunch singularities.

In the present paper we considered a simple cosmological model containing a
mixture of anti-Chaplygin gas and dust. We have shown that the geodesics
equations and their solutions are still well-defined in this case, however
the inclusion of dust generates a nonzero value of the Hubble parameter at
the singularity encounter, generating the following paradox. The dust would
require a continued expansion, which would make the energy density and
pressure of the anti-Chaplygin gas ill-defined. A contraction in turn, would
be compatible with the anti-Chaplygin gas, nevertheless implying an abrupt
change of the Hubble parameter from expansion to contraction. The jump in
the Hubble parameter implies the appearance of the $\delta $ function in the
Raychaudhuri equation (which contains $\dot{H}$).

We have cured this situation by redefining the pressure and energy density
of the anti-Chaplygin gas as distributions. As an equivalent interpretation,
the pressure can be generalized by the addition of a distributional
contribution, while the energy density left unchanged, at the price of
redefining the equation of state of the anti-Chaplygin gas in a
distributional sense. Then all cosmological equations are satisfied in the
same distributional sense. We have also shown, that the Israel junction
conditions are obeyed through the singular spatial hypersurface, in
particular we have enounced and proved a more general form of the Lanczos
equation. The results rely on two Propositions and a Corollary proven in
Appendix A.

The resolution of the paradox at the soft singularity crossing by the
introduction of distributional quantities and equations may look unusual,
however distributional quantities, localized on hypersurfaces are quite
commonly used in general relativity and other gravitational theories.
Spacetime regions are frequently matched by the inclusion of distributional
layers; also shock-waves can be modeled by Dirac $\delta $-functions.
Braneworld models \cite{RS}, \cite{GergelyFriedmann} arise due to the
orbifold boundary conditions, the non-smoothness of the 5-dimensional \
metric at the brane (the jump in its extrinsic curvature) being directly
related to the distributional 3+1 standard model fields embedded in the
5-dimensional spacetime. Besides, metrics allowing distributional curvature
were considered earlier for studying strings and other distributional
sources in general relativity \cite{Geroch}. The applications of the
distributional quantities to the study of Schwarzschild geometry and point
massive particles in general relativity were used in Refs. \cite{Balasin}
and \cite{Katanaev} respectively.

More generically the connection between singularities and the distributional
treatment of the physical quantities is well-known in quantum field theory.
Indeed, the appearance of the ultraviolet divergences can be understood as
the result of the indefiniteness of the product of distributions and the
renormalization procedure could be interpreted as a definition of such a
product \cite{B-Sh}.

We hope that the investigations presented here may turn useful in deriving
similar results in connection with the traversability of other types of
sudden singularities.

While mathematically self-consistent, the scenario presented in this paper
may look somewhat counter-intuitive from the physical point of view. This is
because its essential ingredient is the abrupt change of the expansion into
a contraction. However, such a behavior is not more counter-intuitive that
the absolutely elastic bounce of the ball from a rigid wall, as known in
classical mechanics. Indeed, in the latter case the velocity and the
momentum of the ball change their direction abruptly. That means that an
infinite force acts from the wall onto the ball during an infinitely small
interval of time. The result of this action is however integrable and
results in a finite change of the momentum of the ball. In fact, the
absolutely elastic bounce is an idealization of a process of finite
time-span during which inelastic deformations of both the ball and the wall
are likely. It is reasonable to think that something similar occurs also in
the two-fluid universe model presented in this paper, which undergoes a
transition from an expanding to a contracting phase. The smoothing of this
process should involve some (temporary) geometrically induced change of the
equation of state of matter. Note, that such changes are not uncommon in
cosmology. In the tachyon model \cite{tach0} which was starting point of our
studies of the Big Brake singularities, there was the tachyon -pseudotachyon
transformation driven by the continuity of the cosmological evolution. In a
cosmological model with the phantom field with the cusped potential \cite%
{we-phantom}, the transformations between phantom and standard scalar field
were considered. Thus, one can imagine that the real process of the
transition from the expansion to contraction induced by passing through a
soft singularity can imply some temporary change of the equation of state
which makes the above processes smoother. We hope to explore such a scenario
in the future.

\section*{ACKNOWLEDGMENTS}

We are grateful to V. Gorini, M. O. Katanaev, V. N. Lukash, U. Moschella, D.
Polarski and A. A. Starobinsky for useful discussions. The work of ZK was
supported by OTKA\ grant no. 100216 and AK was partially supported by the
RFBR grant no. 11-02-00643.

\appendix

\section{Propositions on the product and the composition of distributions 
\label{Prop}}

To investigate how the Friedmann universe crosses a soft singularity, we
must solve the field equation in distributional sense. For this purpose we
give the definitions of the product and of the composition of distributions
and prove two propositions. Fisher derived the following result: $\left[
sgn\left( \tau \right) \left\vert \tau \right\vert ^{\lambda }\right] \delta
\left( \tau \right) =0$ for $\lambda >-1$ \cite{Fisherproof}. Our first
proposition generalizes this equation for $\lambda \geq 0$. The second
proposition generalizes Antosik's result: $\left[ 1+\delta \left( \tau
\right) \right] ^{-1}=1$ \cite{Antosik}. Finally, we show a corollary.

Let $\rho \left( \tau \right) $ be any infinitely differentiable function
having the following properties: $i)$ $\rho \left( \tau \right) =0$ for $%
\left\vert \tau \right\vert \geq 1$; $ii)$ $\rho \left( \tau \right) \geq 0$%
; $iii)$ $\rho \left( \tau \right) =\rho \left( -\tau \right) $; $iv)$ $%
\int_{-1}^{1}\rho \left( \tau \right) d\tau =1$. Then $\delta _{n}\left(
\tau \right) =n\rho \left( n\tau \right) $ (with $n=1$, $2$, ...) is a
regular sequence of infinitely differentiable functions converging to Dirac
delta function: $\lim_{n\rightarrow \infty }\left\langle \delta _{n},\varphi
\right\rangle =\left\langle \delta ,\varphi \right\rangle $ for any $\varphi
\in \mathcal{D}$ \cite{Zhi-Fisher}. Here $\mathcal{D}$ denotes the space of
test functions having continuous derivatives of all orders and compact
support. The action of an $f\in \mathcal{D}^{\prime }$ distribution on test
functions $\varphi $ is given by $\left\langle f,\varphi \right\rangle$,
which in the case when $f$ is an ordinary locally summable function is
nothing but $\int_{-\infty }^{\infty }f\left( \tau \right) \varphi \left(
\tau \right) d\tau $. We note that $\delta _{n}\left( \tau \right) $ has the
compact support: $\left( -1/n,1/n\right) $. We will also use the $n$-th
derivative of $f\in \mathcal{D}^{\prime }$ acts as$\ \left\langle df\left(
\tau \right) /d\tau ^{n},\varphi \right\rangle =\left( -1\right)
^{n}\left\langle f\left( \tau \right) ,d\varphi /d\tau ^{n}\right\rangle $.

For an arbitrary distribution $f$, the function $f_{n}\left( \tau \right)
=f\ast \delta _{n}\equiv \left\langle f\left( \tau -x\right) ,\delta
_{n}\left( x\right) \right\rangle $ gives a sequence of infinitely
differentiable functions converging to $f$.

\begin{defn}
: The commutative product of $f$ and $g$ exists and is equal to $h$ on the
open interval $\left( a,b\right) $ ($-\infty \leq a<b\leq \infty $) if%
\begin{equation*}
\lim_{n\rightarrow \infty }\left\langle f_{n}g_{n},\varphi \right\rangle
=\left\langle h,\varphi \right\rangle
\end{equation*}%
for any $\varphi \in \mathcal{D}$ with support contained in the interval $%
\left( a,b\right) $ \cite{Zhi-Fisher} \footnote{%
This definition can be generalized for the cases when the usual limit does
not exist by taking the so-called neutrix limit \cite{Zhi-Fisher}, \cite%
{Corput}. However, we do not need for this more general definition here.}.
\end{defn}

\begin{prop}
: The commutative product of $sgn\left( \tau \right) g\left( \left\vert \tau
\right\vert \right) $ and $\delta \left( \tau \right) $ exists and%
\begin{equation*}
\left[ sgn\left( \tau \right) g\left( \left\vert \tau \right\vert \right) %
\right] \delta \left( \tau \right) =0
\end{equation*}%
for arbitrary $g\left( \left\vert \tau \right\vert \right) $ bounded on
every finite interval.
\end{prop}

\begin{proof}
We would like to show that $\left\langle \left[ sgn\left( \tau \right)
g\left( \left\vert \tau \right\vert \right) \right] \delta \left( \tau
\right) ,\varphi \right\rangle =0$. Using the mean value theorem $\varphi
\left( \tau \right) =\varphi \left( 0\right) +\tau d\varphi \left( \xi \tau
\right) /d\tau $ with $0\leq \xi \leq 1$, we have%
\begin{eqnarray*}
&&\left\vert \left\langle \left[ sgn\left( \tau \right) g\left( \left\vert
\tau \right\vert \right) \right] _{n}\delta _{n}\left( \tau \right) ,\varphi
\right\rangle \right\vert  \\
&\leq &\left\vert \varphi \left( 0\right) \int_{-1/n}^{1/n}\left[ sgn\left(
\tau \right) g\left( \left\vert \tau \right\vert \right) \right] _{n}\delta
_{n}\left( \tau \right) d\tau \right\vert  \\
&&+\sup_{\left\vert \tau \right\vert \leq 1/n}\left\vert \frac{d\varphi
\left( \tau \right) }{d\tau }\right\vert  \\
&&\times \int_{-1/n}^{1/n}\left\vert \tau \left[ sgn\left( \tau \right)
g\left( \left\vert \tau \right\vert \right) \right] _{n}\delta _{n}\left(
\tau \right) \right\vert d\tau \ .
\end{eqnarray*}%
The first integral on the right side of the above equation vanishes because
the integrand is an odd function. For the second integrand, we have%
\begin{eqnarray*}
&&\int_{-1/n}^{1/n}\left\vert \tau \left[ sgn\left( \tau \right) g\left(
\left\vert \tau \right\vert \right) \right] _{n}\delta _{n}\left( \tau
\right) \right\vert d\tau  \\
&=&\int_{-1/n}^{1/n}\left\vert \tau \delta _{n}\left( \tau \right)
\right\vert \int_{-1/n}^{1/n}\left\vert g\left( \left\vert \tau
-x\right\vert \right) \right\vert \delta _{n}\left( x\right) dxd\tau  \\
&\leq &n\sup_{\left\vert \tau \right\vert \leq 1/n}\left\vert \rho \left(
\tau \right) \right\vert \int_{-1/n}^{1/n}\left\vert \tau \delta _{n}\left(
\tau \right) \right\vert \int_{-1/n}^{1/n}\left\vert g\left( \left\vert \tau
-x\right\vert \right) \right\vert dxd\tau  \\
&\leq &2\sup_{\left\vert \tau \right\vert \leq 1/n}\left\vert \rho \left(
\tau \right) \right\vert \sup_{\left\vert \tau \right\vert \leq
1/n}\left\vert g\left( \left\vert \tau \right\vert \right) \right\vert
\int_{-1/n}^{1/n}\left\vert \tau \delta _{n}\left( \tau \right) \right\vert
d\tau  \\
&=&\frac{2}{n}\sup_{\left\vert \tau \right\vert \leq 1/n}\left\vert \rho
\left( \tau \right) \right\vert \sup_{\left\vert \tau \right\vert \leq
1/n}\left\vert g\left( \left\vert \tau \right\vert \right) \right\vert
\int_{-1}^{1}\left\vert y\rho \left( y\right) \right\vert dy \\
&\leq &\frac{2}{n}\sup_{\left\vert \tau \right\vert \leq 1/n}\left\vert \rho
\left( \tau \right) \right\vert \sup_{\left\vert \tau \right\vert \leq
1/n}\left\vert g\left( \left\vert \tau \right\vert \right) \right\vert \ ,
\end{eqnarray*}%
that vanishes for $n\rightarrow \infty $.
\end{proof}

\begin{defn}
: The composition $F\left( f\right) $ of distributions $F$ and $f$ exists
and is equal to $h\in \mathcal{D}^{\prime }$ on the interval $\left(
a,b\right) $ if%
\begin{equation*}
\lim_{n\rightarrow \infty }\left[ \lim_{m\rightarrow \infty
}\int_{a}^{b}F_{n}\left( f_{m}\left( \tau \right) \right) \ \varphi \left(
\tau \right) d\tau \right] =\left\langle h,\varphi \right\rangle
\end{equation*}%
for all $\varphi \in \mathcal{D}$ with support contained in the interval $%
\left( a,b\right) $ \footnote{%
This definition can be generalized for the cases when the usual limit does
not exist by taking double neutrix limit \cite{Fisherdef}, \cite{Ozcag}, 
\cite{FisherKilicman}.}.
\end{defn}

\begin{prop}
: The composition of distribution $P\left( \tau \right) \left[ R\left( \tau
\right) +Q\left( \tau \right) \delta \left( \tau \right) \right] ^{-\omega }$
(where $\omega >0$, $P\left( \tau \right) $ is bounded, $R\left( \tau
\right) \neq 0$, and in some range close $\tau =0$ the signs of $R\left(
\tau \right) $ and $Q\left( \tau \right) $ are the same if $Q\left( \tau
\right) \neq 0$) exists if $P\left( \tau \right) /R^{\omega }\left( \tau
\right) $ exists\footnote{%
We note that this proposition can be held even if $P\left( \tau \right) =1$
and $R\left( \tau \right) =\delta \left( \tau \right) $ with $\omega =1,$ $2,
$ $...$. Indeed, $\delta ^{-\omega }\left( \tau \right) $ exists in neutrix
limit and $\delta ^{-\omega }\left( \tau \right) =0$ \cite{Ozcag}. Thus in
the definition 2, the usual limit must be changed for neutrix limit for this
case.} and 
\begin{equation*}
\frac{P\left( \tau \right) }{\left[ R\left( \tau \right) +Q\left( \tau
\right) \delta \left( \tau \right) \right] ^{\omega }}=\frac{P\left( \tau
\right) }{R^{\omega }\left( \tau \right) }\ .
\end{equation*}
\end{prop}

\begin{proof}
By the definition of composition of distributions, we should calculate 
\begin{eqnarray*}
&&\left\langle \frac{P_{n}\left( \tau \right) }{\left[ R_{m}\left( \tau
\right) +Q_{m}\left( \tau \right) \delta _{m}\left( \tau \right) \right]
_{n}^{\omega }},\varphi \left( \tau \right) \right\rangle  \\
&=&\int_{-\infty }^{\infty }\int_{-1/n}^{1/n}\frac{\varphi \left( \tau
\right) P_{n}\left( \tau \right) \delta _{n}\left( x\right) dxd\tau }{\left[
R_{m}\left( \tau -x\right) +Q_{m}\left( \tau \right) \delta _{m}\left( \tau
-x\right) \right] ^{\omega }}\ .
\end{eqnarray*}%
Performing a change of the variables as $\tau =\tau $, $y=m\left( \tau
-x\right) $, we have%
\begin{eqnarray*}
&=&-\frac{1}{m}\int_{-\infty }^{\infty }\int_{-\infty }^{\infty }\frac{%
\varphi \left( \tau \right) P_{n}\left( \tau \right) \delta _{n}\left( \tau
-y/m\right) }{\left[ R_{m}\left( y/m\right) +mQ_{m}\left( y/m\right) \rho
\left( y\right) \right] ^{\omega }}dyd\tau  \\
&=&-\frac{1}{m}\int_{-\infty }^{\infty }\int_{\Omega _{1}}\frac{\varphi
\left( \tau \right) P_{n}\left( \tau \right) \delta _{n}\left( \tau
-y/m\right) }{R_{m}^{\omega }\left( y/m\right) }dyd\tau  \\
&&-\frac{1}{m}\int_{-\infty }^{\infty }\int_{\Omega _{2}}\frac{\varphi
\left( \tau \right) P_{n}\left( \tau \right) \delta _{n}\left( \tau
-y/m\right) }{\left[ R_{m}\left( y/m\right) +mQ_{m}\left( y/m\right) \rho
\left( y\right) \right] ^{\omega }}dyd\tau \ ,
\end{eqnarray*}%
where $\Omega _{2}=\left\{ y:\left\vert y\right\vert <1\text{ and }\rho
\left( y\right) \neq 0\right\} $ and $\Omega _{1}=%
\mathbb{R}
-\Omega _{2}$. The double limit of the first term is%
\begin{eqnarray*}
&&\lim_{n\rightarrow \infty }\lim_{m\rightarrow \rightarrow \infty }-\frac{1%
}{m}\int_{-\infty }^{\infty }d\tau \varphi \left( \tau \right) P_{n}\left(
\tau \right)  \\
&&\times \int_{\Omega _{1}}\frac{\delta _{n}\left( \tau -y/m\right) }{%
R_{m}^{\omega }\left( y/m\right) }dy \\
&=&\lim_{n\rightarrow \infty }\lim_{m\rightarrow \rightarrow \infty
}\int_{-\infty }^{\infty }d\tau \varphi \left( \tau \right) P_{n}\left( \tau
\right)  \\
&&\times \int_{\substack{ \left\vert x\right\vert <1/n, \\ m\left\vert \tau
-x\right\vert \in \Omega _{1}}}\frac{\delta _{n}\left( x\right) }{%
R_{m}^{\omega }\left( \tau -x\right) }dx \\
&=&\left\langle \frac{P\left( \tau \right) }{R^{\omega }\left( \tau \right) }%
,\varphi \left( \tau \right) \right\rangle \ .
\end{eqnarray*}%
We investigate the absolute value of the second integral. According to our
assumptions for $R$ and $Q$, and since we are interested in $m\rightarrow
\infty $, we can choose $m$ large enough to let the signs of $R$ and $Q$ be
the same, then for $\omega >0$:%
\begin{eqnarray*}
&&\left\vert \frac{1}{m}\int_{-\infty }^{\infty }\int_{\Omega _{2}}\frac{%
\varphi \left( \tau \right) P_{n}\left( \tau \right) \delta _{n}\left( \tau
-y/m\right) }{\left[ R_{m}\left( y/m\right) +mQ_{m}\left( y/m\right) \rho
\left( y\right) \right] ^{\omega }}dyd\tau \right\vert  \\
&\leq &\left\vert \frac{1}{m^{1+\omega }}\int_{-\infty }^{\infty }\varphi
\left( \tau \right) P_{n}\left( \tau \right) \int_{\Omega _{2}}\frac{\delta
_{n}\left( \tau -y/m\right) }{Q_{m}\left( y/m\right) \rho ^{\omega }\left(
y\right) }dyd\tau \right\vert \ .
\end{eqnarray*}%
Performing a change of the variables as $z=n\left( \tau -y/m\right) $, $y=y$%
, we have%
\begin{eqnarray*}
&\leq &\frac{1}{m^{1+\omega }}\int_{-1}^{1}\int_{\Omega _{2}}\left\vert
\varphi \left( \frac{z}{n}+\frac{y}{m}\right) P_{n}\left( \frac{z}{n}+\frac{y%
}{m}\right) \frac{\rho \left( z\right) }{\rho ^{\omega }\left( y\right) }%
\right\vert dydz \\
&\leq &\frac{1}{m^{1+\omega }}\sup_{\Omega _{2},\left\vert z\right\vert \leq
1}\left\vert \varphi \left( \frac{z}{n}+\frac{y}{m}\right) P_{n}\left( \frac{%
z}{n}+\frac{y}{m}\right) \rho ^{-\omega }\left( y\right) \right\vert  \\
&&\times \int_{-1}^{1}\rho \left( z\right) dz\int_{-1}^{1}dy \\
&=&\frac{2}{m^{1+\omega }}\sup_{\Omega _{2},\left\vert z\right\vert \leq
1}\left\vert \varphi \left( \frac{z}{n}+\frac{y}{m}\right) P_{n}\left( \frac{%
z}{n}+\frac{y}{m}\right) \rho ^{-\omega }\left( y\right) \right\vert \ ,
\end{eqnarray*}%
that vanishes for $m\rightarrow \infty $ if $P$ is bounded.
\end{proof}

\begin{cora}
: The distribution $d\left\{ P\left( \tau \right) \left[ R\left( \tau
\right) +Q\left( \tau \right) \delta \left( \tau \right) \right] ^{-\omega
}\right\} /d\tau $ (with the same properties for $P$, $R$, $Q$ and $\omega $
as in proposition 2) exists if $P\left( \tau \right) /R^{\omega }\left( \tau
\right) $ and its derivative exist, and%
\begin{equation*}
\frac{d}{d\tau }\frac{P\left( \tau \right) }{\left[ R\left( \tau \right)
+Q\left( \tau \right) \delta \left( \tau \right) \right] ^{\omega }}=\frac{d%
}{d\tau }\frac{P\left( \tau \right) }{R^{\omega }\left( \tau \right) }\ .
\end{equation*}
\end{cora}

\begin{proof}
Applying the derivative of a distribution at tests functions, and using the
fact that $d\varphi /d\tau \in \mathcal{D}$ for any $\varphi \in \mathcal{D}$%
, and by the proposition 2, we have 
\begin{eqnarray*}
&&\left\langle \frac{d}{d\tau }\frac{P\left( \tau \right) }{\left[ R\left(
\tau \right) +Q\left( \tau \right) \delta \left( \tau \right) \right]
^{\omega }},\varphi \right\rangle  \\
&=&-\left\langle \frac{P\left( \tau \right) }{\left[ R\left( \tau \right)
+Q\left( \tau \right) \delta \left( \tau \right) \right] ^{\omega }},\frac{d%
}{d\tau }\varphi \right\rangle  \\
&=&-\left\langle \frac{P\left( \tau \right) }{R^{\omega }\left( \tau \right) 
},\frac{d}{d\tau }\varphi \right\rangle =\left\langle \frac{d}{d\tau }\frac{%
P\left( \tau \right) }{R^{\omega }\left( \tau \right) },\varphi
\right\rangle \ .
\end{eqnarray*}
\end{proof}

\section{Two simple examples of the product and of the decomposition of
distributions}

The definition of the product and the composition of distributions, used in
this paper and presented in the Appendix A are not often encountered in
physics. Thus, to give the reader some flavor of the corresponding
considerations, using simpler means we decided to give two semi-heuristic
examples of such products and compositions. We consider first a remarkable
formula 
\begin{equation}
\mathcal{P}\left( \frac{1}{x}\right) \delta (x)=-\frac{1}{2}\delta ^{\prime
}(x)~,  \label{dist-eq}
\end{equation}%
which was first proven in \cite{Mikus}. Here $\mathcal{P}$ means the
principal value of the corresponding function. We shall prove here that the
regularizing succession of functions with compact support $\rho $, employed
in the Appendix A and the references therein, can be chosen alternatively as
the family of the Cauchy-Lorentz functions 
\begin{equation}
f_{\epsilon }(x)=\frac{1}{\pi }\frac{\epsilon }{x^{2}+\epsilon ^{2}}~.
\label{bell}
\end{equation}%
It is well known that when the small parameter $\epsilon \rightarrow 0$, the
functions of this family tend in the distributional sense to the Dirac $%
\delta $ function. Obviously, the convolution of the function (\ref{bell}
with the Dirac $\delta $ function gives again the same function (\ref{bell}%
): 
\begin{equation}
f_{\epsilon }\ast \delta (x)=f_{\epsilon }(x)~.  \label{convol}
\end{equation}%
The calculation of the convolution of the principal value $\mathcal{P}\left( 
\frac{1}{x}\right) $ with the function $f_{\epsilon }(x)$ is slightly more
complicated: 
\begin{eqnarray}
&&\mathcal{P}\left( \frac{1}{x}\right) \ast f_{\epsilon
}(x)=\lim_{\varepsilon \rightarrow 0}\left( \int_{-\infty }^{x-\varepsilon
}dy\frac{1}{x-y}\frac{\epsilon }{\pi (y^{2}+\epsilon ^{2})}\right.  \notag \\
&&\left. +\int_{x+\varepsilon }^{\infty }dy\frac{1}{x-y}\frac{\epsilon }{\pi
(y^{2}+\epsilon ^{2})}\right) =\frac{x}{x^{2}+\epsilon ^{2}}~.
\label{convol1}
\end{eqnarray}%
The product of the expressions (\ref{convol}) and (\ref{convol1}) is 
\begin{equation}
\mathcal{P}\left( \frac{1}{x}\right) _{\epsilon }\ast \delta _{\epsilon }(x)=%
\frac{\epsilon x}{\pi (x^{2}+\epsilon ^{2})^{2}}~.  \label{convol-prod}
\end{equation}%
Let us now consider a family of functions 
\begin{equation}
\frac{df_{\epsilon }(x)}{dx}=-\frac{2x\epsilon }{\pi (x^{2}+\epsilon
^{2})^{2}}~.  \label{bell-der}
\end{equation}%
One can easily prove that if the family of functions (\ref{bell}) converges
in the distributional sense to the Dirac $\delta $ function, the family of
their derivatives (\ref{bell-der}) converges to the derivative of the delta
function. Now, comparing the right-hand sides of Eqs. (\ref{bell-der}) and (%
\ref{convol-prod}) we see that when $\epsilon \rightarrow 0$ the product in
the left-hand side of Eq. (\ref{convol-prod}) converges in the
distributional sense to $-\frac{1}{2}\delta ^{\prime }(x)$ and thus the
correctness of the equality (\ref{dist-eq}) is checked.

Now let us discuss the Antosik identity \cite{Antosik} 
\begin{equation}
\frac{1}{1+\delta (x)}=1.  \label{Ant}
\end{equation}%
Here we have the composition of the distributions $F(g)$, where $F=\frac{1}{g%
}$ and $g(x)=1+\delta (x)$. Calculating the convolutions of the
distributions $F$ and $g$ with the Cauchy-Lorentz functions (\ref{bell}) we
obtain 
\begin{equation}
F_{\sigma }(g)=F\ast f_{\sigma }(g)=\frac{g}{g^{2}+\sigma ^{2}},
\label{comp}
\end{equation}%
\begin{equation}
g_{\epsilon }(x)=1+\frac{\epsilon }{x^{2}+\epsilon ^{2}}.  \label{comp1}
\end{equation}%
Correspondingly the composition of these functions is 
\begin{equation}
F_{\sigma }(g_{\epsilon })=\frac{1+\frac{\epsilon }{\pi (x^{2}+\epsilon ^{2}}%
}{\sigma ^{2}+\left( 1+\frac{\epsilon }{\pi (x^{2}+\epsilon ^{2}}\right) ^{2}%
}  \label{comp2}
\end{equation}%
and it is easy to check that 
\begin{equation}
\lim_{\sigma \rightarrow 0}\lim_{\epsilon \rightarrow 0}F_{\sigma
}(g_{\epsilon })=1,
\end{equation}%
confirming the identity (\ref{Ant}).

\section{The Lanczos equation\label{Lanczos}}

For a generic junction surface the Lanczos equation emerges from the
Gauss-Codazzi relations \cite{Gregory}, \cite{GergelyFriedmann}. The
projected Lie derivative of the extrinsic curvature $K_{ab}$ in the normal
direction $n$ to the surface is 
\begin{equation}
h_{a}^{i}h_{b}^{j}\mathcal{L}_{\mathbf{n}}K_{ij}=-3\epsilon\left(
h_{a}^{i}h_{b}^{k}T_{ik}-\frac{h_{ab}}{2}g^{ik}T_{ik}\right) +\mathcal{Z}%
_{ab}\   \label{extrderiv}
\end{equation}%
(Eq. (21) of \cite{GergelyFriedmann} in the units $8\pi G/3=1$), with%
\begin{eqnarray}
\mathcal{Z}_{ab} &=&-\epsilon \mathcal{R}%
_{ab}+2K_{ac}K_{b}^{c}-g^{ik}K_{ik}K_{ab}  \notag \\
&&+D_{b}\alpha _{a}-\epsilon \alpha _{b}\alpha _{a}\ .
\end{eqnarray}%
Here $g_{ab}$ is the space-time metric, $h_{ab}=g_{ab}-\epsilon n_{a}n_{b}$ (%
$\epsilon =n^{a}n_{a}=\left\{ -1,1\right\} $) is the induced metric on the
junction surface, and $T_{ab}$ is the energy-momentum tensor. The tensor $%
\mathcal{Z}_{ab}$ depends only on geometrical quantities: $\mathcal{R}_{ab}$
and $D$ are the Ricci tensor and covariant derivative induced on the
hypersurface, and $\alpha _{a}=n^{c}\nabla _{c}n_{a}$, with $\nabla $ the
4-dimensional covariant derivative.

When the energy-momentum tensor is a \textit{sum} $T_{ik}=\Pi _{ik}+\Upsilon
_{ik}\delta \left( \tau \right) $ (where $\tau $ is the coordinate adapted
to $n$, i.e. $n=t_{S}^{-1}\partial /\partial \tau $, and $n^{a}\Upsilon
_{ab}=0$), with $\Pi _{ik}$ the regular 4-dimensional part and $\Upsilon
_{ik}$ the distributional part on the hypersurface, integration of Eq. (\ref%
{extrderiv}) across $\tau $ through an infinitesimal range containing the
hypersurface keeps only the distributional part, leading to the Lanczos
equation \cite{Lanczos}, \cite{Israel}. 
\begin{equation}
\Delta K_{ab}=-3\epsilon \left( \Upsilon _{ab}-\frac{\Upsilon }{2}%
h_{ab}\right) \ ,
\end{equation}%
or equivalently 
\begin{equation}
-3\epsilon \Upsilon _{ab}=\Delta K_{ab}-h_{ab}\Delta K\ .
\end{equation}%
Here $\Upsilon $ is the trace of $\Upsilon _{ab}$. As $\mathcal{Z}_{ab}$ is
finite, its contribution to the integral across the infinitesimal range also
vanishes. Without a distributional energy-momentum part, the extrinsic
curvature should be continuous.

Let us now specialize this for a junction along a maximally symmetric $\tau
=0$ spacelike hypersurface (a hyperplane with $\mathcal{R}_{ab}=0$) embedded
in a flat Friedmann space-time. The normal vector $n$ has zero acceleration $%
\alpha _{a}=0$ and the extrinsic curvature becomes $K_{ab}=\dot{a}a%
\widetilde{h}_{ab}$, with $\widetilde{h}_{ab}$ the 3-dimensional Euclidean
metric. The curvature term is $\mathcal{Z}_{ab}=-H^{2}a^{2}\widetilde{h}%
_{ab} $ and the energy momentum tensors are $\Pi _{ab}=\widetilde{\rho }%
n_{a}n_{b}+\widetilde{p}a^{2}\widetilde{h}_{ab}$ and $\Upsilon _{ab}=%
\overline{p}a^{2}\widetilde{h}_{ab}$. Since the projected Lie-derivative in
Eq. (\ref{extrderiv}) becomes a time derivative, the equation reads%
\begin{equation}
\frac{\partial }{\partial t}\left( Ha^{2}\right) =\left\{ -H^{2}+\frac{3}{2}%
\left[ \widetilde{\rho }-\widetilde{p}-\overline{p}\delta \left( \tau
\right) \right] \right\} a^{2}\ ,  \label{Hderiv}
\end{equation}%
which is a combination of the Raychaudhuri and Friedmann equations. For
finite $H$, $\widetilde{\rho }$ and $\widetilde{p}$ as before the
integration of Eq. (\ref{Hderiv}) across an infinitesimal time range $\tau $
leads to the Lanczos equation

\begin{equation}
\Delta H|_{t_{s}}=-\frac{3}{2}\overline{p}\ .  \label{Hjump}
\end{equation}

\end{document}